\documentclass[twocolumn,amsthm]{autart}  
\usepackage{float}
\usepackage{amsmath}
\usepackage{amsthm}
\usepackage{amssymb}
\usepackage{graphicx}
\usepackage{cite}
\usepackage{xcolor}
\newtheorem{theo}{Theorem}
\newtheorem{ass}{Assumption}
\newcommand{\R}{\mathbb{R}}
\newcommand{\F}{\mathcal{F}}
\makeatother

\usepackage{babel}
\usepackage{subfig}

\begin{document}
	
	\begin{frontmatter}

		\title{Discrete implementations of sliding-mode controllers with barrier-function adaptations require a revised framework\thanksref{footnoteinfo}}		
		\thanks[footnoteinfo]{This paper was not presented at any IFAC meeting.}

		\author[Mex]{Luis Ovalle}\ead{luis.ricardo.ovalle@gmail.com}, % e-mail address 
		\author[Rome]{Andres Gonzalez}\ead{andresglezrod29@gmail.com }, % 
		\author[Rome]{Leonid Fridman}\ead{lfridman@unam.mx}, %
		\author[arg]{Hernan Haimovich}\ead{haimovich@cifasis-conicet.gov.ar}, %
		\address[Mex]{TecNM/Instituto Tecnologico Superior de Lerdo, Av. Tecnologico No. 1555 Sur, Ciudad Lerdo, Durango.}             % full addresses
		\address[Rome]{Universidad Nacional Aut\'{o}noma de M\'{e}xico (UNAM), Facultad de Ingenier\'{i}a, 04510, Mexico City, Mexico.}             % full addresses
		\address[arg]{Centro Internacional Franco-Argentino de Ciencias de la Informaci\'{o}n y de Sistemas (CIFASIS) CONICET-UNR Rosario, Argentina}   % here.
			%	\address[Baiae]{Department of Information Technologies and AI, Sirius University of Science and Technology, Sochi, Russia}
		\begin{abstract}                          % Abstract of not more than 200 words.
Challenges in the discrete implementation of sliding-mode controllers (SMC) with barrier-function-based adaptations are analyzed, revealing fundamental limitations in conventional design frameworks. It is shown that under uniform sampling, the original continuous-time problem motivating these controllers becomes theoretically unsolvable under standard assumptions. To address this incompatibility, a revised control framework is proposed, explicitly incorporating actuator capacity constraints and sampled-data dynamics. Within this structure, the behavior of barrier function-based adaptive controllers (BFASMC) is rigorously examined, explaining their empirical success in digital implementations. A key theoretical result establishes an explicit relation between the actuator capacity, the sampling rate, and the width of the barrier function, providing a principled means to tune these controllers for different application requirements. This relation enables the resolution of various design problems with direct practical implications. A modified BFASMC is then introduced, systematically leveraging sampling effects to ensure finite-time convergence to a positively invariant predefined set, a key advancement for guaranteeing predictable safety margins. 
		\end{abstract}
		
	\end{frontmatter} 

%\maketitle
\section{Introduction}
Controllers ensuring predefined performance, such as predefined performance controllers (PPC) \cite{PPC}, monitoring-function-based controllers \cite{rodrigues2022adaptive}, and barrier function-based adaptive controllers (BFASMC) \cite{BF1}, have gained popularity in the literature. This paradigm confines system behavior within a predefined set \textit{a priori}, regardless of external disturbances, defining what we term the predefined performance problem. Critically, achieving this objective under unknown perturbation bounds needs adaptation mechanisms that do not rely on prior disturbance knowledge.

While classical adaptive sliding-mode controllers (ASMCs) require assumptions on the upper bound of the perturbation or its derivatives to ensure stability \cite{Plestan,shtessel2}, BFASMC uniquely circumvents this limitation. By using barrier functions with vertical asymptotes at the predefined set boundaries, BFASMC guarantees confinement without \textit{a priori} knowledge of perturbation bounds, a capability unmatched by gain-increasing or equivalent-control-based ASMCs.

A major criticism of BFASMC is the requirement for gains that grow unboundedly as state trajectories approach the boundary of the predefined set. However, this characteristic is essential to ensure confinement within the predefined set despite bounded perturbations.

In real-world applications, predefined performance control must be achieved with actuators of finite, \textit{a priori} known capacity. Thus, conditions must exist to prevent system trajectories from approaching the boundary arbitrarily closely. For instance, \cite{TAC} proves that BFASMC not only enforces a barrier but also maintains trajectories within a positively invariant \textit{final set}, ensuring a minimum distance from the boundary. This property must hold for all controllers of this type.

This paper demonstrates that, under standard theoretical conditions, the predefined performance problem cannot be solved under sampling, irrespective of the controller chosen.

To contextualize BFASMC, note that it originates from adaptive sliding modes, where three classical design methods exist \cite{Plestan,shtessel2}:

\begin{itemize}
\item \textbf{Increasing gains} \cite{Moreno}: Controller gains increase until disturbances are compensated, ensuring sliding mode establishment. However, this approach suffers from gain overestimation, leading to chattering, and uncertainty regarding the precise moment when sliding mode is achieved.
\item \textbf{Reconstruction of equivalent control} \cite{utkin,edwards}: A filter estimates the equivalent control, allowing gain adjustments based on the estimated disturbance. This requires upper bounds on perturbations and their derivatives. While second-order sliding modes improve asymptotic precision, chattering remains unavoidable \cite{boiko}.
\item \textbf{Increasing and decreasing gains} \cite{ID2,Gian}: Gains increase until sliding mode is reached and decrease until it is lost, resulting in ultimate boundedness. However, the ultimate bound depends on unknown perturbation bounds, and the precise moment of reaching the bound remains uncertain. Chattering persists due to the enforced sliding mode.
\end{itemize}

A modern perspective on adaptive sliding mode design should simultaneously satisfy three key properties:
\begin{enumerate}
\item Avoid chattering.
\item Allow predefined ultimate performance.
\item Prevent gain overestimation for disturbance compensation.
\end{enumerate}

BFASMC, as proposed in \cite{BF1}, ensures that state trajectories remain within a predefined barrier set of fixed size. The gain design methodology employs continuous concave functions with vertical asymptotes at the barrier width. Since the control law vanishes at the origin, sliding mode is not reached, and chattering does not occur.

While BFASMC has seen success in applications including robot manipulators \cite{BFMAN}, surface vehicles \cite{BFNNSV}, ABS systems \cite{rodrigues2022adaptive}, Duffing oscillators \cite{mousavi2023barrier}, linear motors \cite{BFLM}, and piezoelectric actuators \cite{BFPZE}, among others, the problem formulation in \cite{BF1} considers unit relative degree systems, so reported applications utilize appropriate output functions to reduce the problem to a scalar system.

Despite the promising results observed in these applications, the implementation of BFASMC, as presented in \cite{BF1}, implicitly assumes two main properties: an Infinitely fast response capability and an infinite control authority. Neither assumption holds in practical applications. Existing implementations implicitly rely on analog design assumptions, but sample-and-hold schemes introduce open-loop behavior between sampling instants. This requires a revised framework that formally reconciles barrier-function adaptations with discrete-time dynamics.

This paper shows that, under standard theoretical conditions, the predefined performance problem cannot be solved in a sampled-data setting, regardless of the controller employed. This is a fundamental issue, as sampling disrupts the feedback mechanisms that ensure predefined performance in continuous-time implementations. To address this challenge, we introduce a revised theoretical framework that explicitly accounts for the interaction between sampling, actuator constraints, and barrier function-based adaptation. Within this framework, we establish a fundamental relation between the maximal admissible perturbation, the sampling rate, and the width of the barrier function, which allows tuning of the controller to solve different practical problems. This insight provides a theoretical explanation for the success of reported BFASMC applications, despite their apparent incompatibility with sampling constraints.

Unlike previous works on discrete-time adaptive sliding-mode control, this paper provides a fundamental impossibility result and a revised problem formulation, establishing a rigorous framework for sampled-data BFASMCs.

For the sake of clarity, the key contributions are summarized as follows:

\begin{itemize}
\item A formal definition of the predefined performance problem.
\item Proof of the unsolvability of the predefined performance problem under sampling.
\item Proposal of a sampled version of the predefined performance problem.
\item Study of the BFASMC solution under sampling, demonstrating preservation of its three main features.
\item Design of a modified BFASMC ensuring finite-time convergence to the predefined ultimate bound in sampled controllers.
\end{itemize}
%To the best of the authors knowledge this note presents for the first time a study dealing with the sample-and-hold implementation of a predefined performance controller of any kind.

\section{Motivation Example}
In this section BFASMC design will be employed to contextualize the problem of predefined performance control.

\subsection{Considered System}
Consider a system of the form:
\begin{equation}\label{eq:sysprem}
\dot{x}=g(t,x)u+\zeta(t,x)
\end{equation}
where $x(t),u(t)\in\R$ are the state and control input values at time $t$, respectively, $g(t,x)$ is an uncertain input gain and $\zeta(t,x)$ is regarded as a disturbance. The main objective of predefined performance controllers is to ensure that the constraint $|x(t)|<\varepsilon$ is fulfilled for some $\varepsilon>0$, for all $t\geq t_0$, regardless of the effect of external disturbances\footnote{If the problem to be solved is not a scalar system of the form \eqref{eq:sysprem}, it is possible to consider an appropriate output function with relative degree one such that the behavior of the output is predefined, \textit{e.g.} a sliding variable.}. The set $\mathcal{E}=\{x\in\mathbb{R}:|x|\leq\varepsilon\}$ is called the predefined set.

\begin{ass}
  \label{ass:sys} System \eqref{eq:sys} satisfies the following
	\begin{itemize}
		\item The function $\zeta(t,x)$ is continuous with respect to $x$ for fixed $t$ and Lebesgue measurable with respect to $t$ for fixed $x$. Furthermore, there exists some (possibly unknown) positive constant $\bar \zeta\geq 0$ such that $|\zeta(t,x)|<\bar \zeta$ for all $t\geq t_0$ and $|x|\leq \varepsilon$.
		\item There exist two positive constants $g_1$ and $g_2$ such that $0 < g_1\leq g(t,x)\leq g_2$ for all $t>t_0$ and $|x|\leq \varepsilon$.
\end{itemize}\end{ass}
The lower bound of second item of Assumption \ref{ass:sys} implies that it is possible to write \eqref{eq:sysprem} as:
\begin{equation}\label{eq:sys}
	\dot{x}=g(t,x)\left(\delta(t,x)+u \right)
\end{equation}
with $\delta(t,x)=\zeta(t,x)/g(t,x)$. If \eqref{eq:sysprem} fulfills Assumption \ref{ass:sys}, so does \eqref{eq:sys}, for some unknown $\bar\delta=\bar\zeta/g_1$.

\subsection{Barrier-Function Adaptive Control}
The objective of BFASMC is to design a controller which, simultaneously: 1) Constrains the state to a predefined neighborhood of the origin 2) Avoids the chattering phenomenon 3) Avoids knowledge of the bound of the perturbation and its derivative. One BFASMC is the control law presented in \cite{BF1}:
\begin{equation}\label{eq:BFA}
	u = \kappa(x) := -\frac{1}{\varepsilon-|x|}x,
\end{equation}
where $\varepsilon>0$ is a design parameter which defines the so-called barrier set $|x|<\varepsilon$.

The following assumption is standard in BFASMC design.
\begin{figure}\begin{center}
		\includegraphics[width=0.9\linewidth]{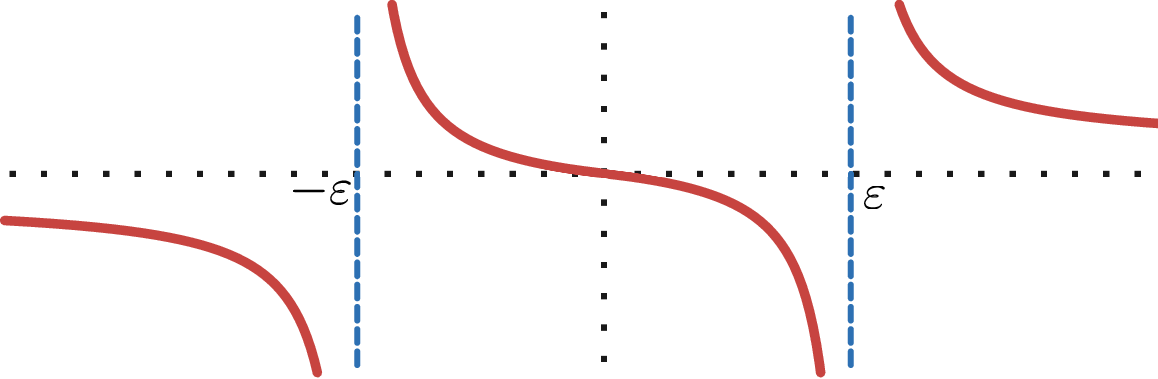}
	\end{center}
	\caption{Graphical representation of the BFASMC}\label{fig:BFA}
\end{figure}

\begin{ass}\label{ass:icc}
$x(t_0)\in\left[-\beta\varepsilon,\beta\varepsilon\right]$, with $\beta\in(0,1)$.
\end{ass}
\begin{rem}
	Control design employing BFASMC involves two stages: convergence to a small region and keeping the state within the region thereafter. The BFASMC is employed for the second stage, and therefore Assumption \ref{ass:icc} has become standard in BFASMC design. In regards to the first stage, adaptation mechanisms, which ensure convergence, always exist, such as the algorithm proposed in \cite{Diego}, which ensures convergence to an arbitrary neighborhood of the origin in a predefined time.
\end{rem}
Figure \ref{fig:BFA} shows a plot of \eqref{eq:BFA} as a function of the state. The main problem is that the sign of the control effort changes once the state leaves the barrier set. In an actual implementation, if the predefined ultimate behavior is not ensured, blind application of such a control law might cause instability.
%Thus, it is necessary to ensure that the predefined ultimate behavior can be achieved regardless of the perturbation. 
The strategy presented in \cite{BF1} could be modified easily by switching to some adaptive law whenever the state leaves the predefined region. However, the occurrence of this event would cause the loss of performance guarantees. Furthermore, this issue needs to be considered if one is to properly explain the cited applications, which show successful application of the results in \cite{BF1} without modifying the control law.

Then, considering this set of conditions, the next theorem (taken from \cite{BF1,TAC}) describes the properties of the BFA:
\begin{theo}
Let system \eqref{eq:sys} be controlled by \eqref{eq:BFA}, and Assumptions \ref{ass:sys} and \ref{ass:icc} hold. Then for all $t\geq t_0$, $|x(t)|\leq\max\left\{\beta,\frac{\varphi}{\varphi+1}\right\}\varepsilon<\varepsilon$ with $\varphi = \bar\delta + \theta$ and some $\theta\in(0,1)$. Moreover, there exists a $T \geq 0$ such that for all  $t\geq t_0 + T$, $|x(t)|$ belongs to the set $S_{\frac{\varphi}{\varphi+1}\varepsilon}:= \left\{x:|x|\leq \frac{\varphi}{\varphi+1}\varepsilon\right\}$.
\end{theo}
This motivates the following problem formuation:
\begin{prob}\label{prob:PPP} (Predefined performance problem, PPP) Given some arbitrary  positive constant $\varepsilon$, a control law $u=\eta(t,x)$ is said to solve the PPP for system~(\ref{eq:sys}) if for every perturbation bound $\bar\delta > 0$ there exists $\theta<1$, such that for every $|x(t_0)|\leq \theta\varepsilon$ the set $\F=\left\{x\in\R:|x|\leq\theta\varepsilon\right\}$ is positively invariant for every possible $g(t,x)$ and $\delta(t,x)$ satisfying $|\delta(t,x)| \le \bar\delta$.
\end{prob}

%\begin{rem}\label{rem:PPP}
%Any controller seeking to solve the PPP must be able to exert an unbounded amount of force to compensate any bounded perturbation. For sampled systems, the controller is left to operate in open-loop between samples; then, as the controller cannot be modified in-between samples, there always exists a bounded perturbation such that any predefined behavior is lost, regardless of any sampling rate, \textit{i.e.} \underline{the PPP cannot be solved under sampling}.
%\end{rem}
\begin{rem}\label{rem:PPPex}
%Problem \ref{prob:PPP} deals not only with the formal definition of predefined performance, but also with the main criticisms of this kind of controllers:
%
  In order to have a control law solve the PPP, a controller must be able to keep the state bounded by $\varepsilon$ irrespective of the effect of \underline{any} bounded perturbation. Therefore, if the norm of the state, $|x|$, approaches the predefined-performance value $\varepsilon$, the control action must tend to infinity. However, any bounded perturbation will be compensated by a bounded control action, since $\mathcal{F}\subset\mathcal{E}$ is rendered positively invariant; thus, $\mathcal{F}$ will be referred to as the final set.   %, since the PPP actually requires the existence of a so-called final set which is positively invariant and lies within the predefined set. Thus, the distance from the final set to the boundary of the predefined set is always bounded away from zero. This means that the compensation of any bounded disturbance can be achieved with a bounded control.
The existence of the final set actually gives some leeway in the sense of ensuring that some level of ignored uncertainties, such as measurement noise, can be tolerated, albeit proper further analysis is still needed.
\end{rem}
%\subsection{Sampling and Barrier Function Adaptation}

%In the sequel, it will be assumed that $u$ is implemented by means of a sample-and-hold method under uniform sampling, i.e. $$u(t,x(t))=u\left(\tau\left\lfloor\frac{t-t_0}{\tau} \right\rfloor,x\left(\tau\left\lfloor\frac{t-t_0}{\tau} \right\rfloor\right)\right),$$ for all $t\geq t_0$, where $\tau$ represents some constant sampling rate and $\lfloor\cdot\rfloor$ represents the floor function.

\subsection{Unsolvability of the PPP under Sampling}

It will be shown that the PPP, Problem \ref{prob:PPP}, cannot be solved under sampling. Let $\{t_k\}_{k=0}^\infty$ be any increasing sequence of sampling instants. Under zero-order hold, one has
\begin{align}
  u(t) = u_k\quad t \in [t_k,t_{k+1})
\end{align}
where $u_k$ is the constant control action applied between the sampling instants $t_k$ and $t_{k+1}$. In this context, defining $x_k = x(t_k)$, the solution to~(\ref{eq:sys}) satisfies
\begin{equation}\label{eq:CLC3}
\begin{array}{rl}
	x_{k+1}&=x_k + \int_{t_k}^{t_k+1}g(s,x(s))\left(u_k+\delta(s,x(s))ds\right)
\end{array}
\end{equation}
%where $u_k$ is a sample of some $u(t,x(t))$, at $t-t_0=k\time \tau$.
\begin{prop}\label{prop:del} Let $0\le t_0 < t_1$ and $|x_0|\le\theta\varepsilon$ for $\theta\in[0,1)$. Let $u_0$ denote the control action computed at $t_0$, to be applied through zero-order hold in the interval $[t_0,t_1)$. Then, there exists some bounded $\delta(t,x)$ such that according to \eqref{eq:CLC3}, $x_1\not\in(-\varepsilon,\varepsilon)$.
\end{prop}
\begin{proof}
%	First, note that, for any $\tau\geq \rho$, a constant disturbance $\delta(t,x)\equiv\frac{\varepsilon}{\rho}$, one would have $\int_{0}^{\tau}\delta(s,x(s))ds=\varepsilon\frac{\tau}{\rho}\geq \varepsilon$.
	
%Assume $x_0\geq0$. Then, $$x_1-x_0=\int_{0}^{\tau}g(s,x(s))\left(u_k+\delta(s,x(t))\right)ds.$$ 
Let us assume $\delta(t,x)=-\varpi$, for some $\varpi>u_0$ such that $x_1-x_0=-\left(\varpi-u_0\right)\int_{t_0}^{t_1}g(s,x(s)<0$, which means that $x_1<x_0$. Since $|x_0|\leq\theta\varepsilon$, the only possibility for $x_1\not\in(-\varepsilon,\varepsilon)$ is $x_1<-\varepsilon$. Define $\tau=t_1-t_0$, it suffices $\frac{\varepsilon}{g_1\tau}(1+\theta)+ u_0\leq\varpi$, which follows from $g(t,x(t))\geq g_1$ and is feasible if $\varpi\geq\frac{\varepsilon}{g_1\tau}(1+\theta)+u_0$. Choosing $\varpi\geq \frac {2\varepsilon}{g_1\tau}+u_0$,  $x_1\not\in(-\varepsilon,\varepsilon)$. Any $\delta(t,x)=-\rho\varpi$, with $\rho\geq 1$, leads to $x_1\not\in(-\varepsilon,\varepsilon)$.\end{proof}

%\begin{rem}
Since~(\ref{eq:CLC3}) expresses the solution of (\ref{eq:sys}) at two consecutive sampling instants, Proposition~\ref{prop:del} shows that without prior knowledge of a bound for the perturbation, $x(t)\in\mathcal{E}$ cannot be ensured under sampling, regardless of the complexity of the control law. Therefore, the PPP cannot be solved \underline{by any controller} under any form of sampling. %\end{rem}
%
%This proposition shows that, under the considered conditions, it is not possible to solve the PPP, since the compensation of \textit{any} bounded perturbation is required.
The impossibility of solving the PPP under sampling implies that the formal problem to be solved needs to be modified to reconcile the success of the practical applications of BFASMC reported in the literature with the theoretical analysis employed. The next section aims to propose a theoretical context in which BFASMC can be studied with the aim to explain, from a theoretical point-of-view, the precise guarantees that the reported experimental results have.

\section{Predefined Performance Problem under Sampling}
The main unrealistic feature of the PPP is the fact that the control must be able to compensate every bounded perturbation, irrespective of its (possibly unknown) bound, $\bar\delta$. In reality, the control authority of any system is finite. As a consequence, the maximum bound for perturbations that can be successfully compensated is predetermined by the actuator capacity of the system or, from the control system design point of view, the actuator would have to be selected based on knowledge of some perturbation bounds. Moreover, any practical implementation should account for sampling.

Since uniform sampling, \textit{i.e.} sampling that occurs at equally spaced time instants, is the type employed in the reported applications of BFASMC, in the sequel it will be assumed that $u$ is implemented by means of a sample-and-hold method under uniform sampling,\textit{ i.e.}
$ u(t)=u\left(\tau\left\lfloor\frac{t-t_0}{\tau}\right\rfloor\right),$
%,x\left(\tau\left\lfloor\frac{t-t_0}{\tau} \right\rfloor\right)\right),$$
for all $t\geq t_0$, where $\tau$ represents some constant sampling period and $\lfloor\cdot\rfloor$ represents the floor function, \textit{i.e.} the greatest integer not greater than its argument.
In this practical implementation context, a reasonable modification of the PPP is the following:
%Then, if a practical implementation of the controller is considered, which is the main motivation to study sampled systems, it makes sense to modify the PPP in the following manner:
\begin{prob} \label{prob:PPPS}(PPP under uniform sampling, PPPS) A control law $u=\eta(t,x)$ is said to solve the PPPS for system~(\ref{eq:sys}) with predefined ultimate bound $\varepsilon$, maximum admissible perturbation $\bar\delta$ and sampling period $\tau$, if for every perturbation bound $\bar\vartheta \in (0, \bar\delta)$ there exists some $\theta<1$, such that for every $|x(t_0)|\leq \theta\varepsilon$ the set $\F=\left\{x\in\R:|x|\leq\theta\varepsilon\right\}$ is rendered positively invariant for every possible $g(t,x)$ and $\delta(t,x)$ satisfying $|\delta(t,x)|\le \bar\vartheta$, when the control law is implemented under uniform sampling and zero-order hold with a sample rate of $\tau$ as:
  $ u(t)=\eta\left(\tau\left\lfloor\frac{t-t_0}{\tau}\right\rfloor,
  x\left(\tau\left\lfloor\frac{t-t_0}{\tau} \right\rfloor\right)\right)
  \quad \forall t\ge t_0.$
\end{prob}

%\begin{rem}
  Unlike the PPP, the PPPS does not require the controller to become unbounded because the maximum admissible value of the perturbation, namely $\bar\delta$, is assumed to be known.

The consideration of the PPPS motivates the following modification of Assumption \ref{ass:sys}:
	\begin{ass}\label{ass:sys1} System \eqref{eq:sys} satisfies:
		\begin{itemize}
			\item There exist two positive constants $g_1$ and $g_2$ such that $0 < g_1\leq g(t,x)\leq g_2$ for all $t>t_0$ and $|x|\leq \varepsilon$.
			\item There exists some positive constant, $c_1>0$, such that the control signal is upper bounded as $|u(t)|\leq c_1$.
			\item There exists some $\bar \delta<c_1$ such that $|\delta(t,x)|\leq \bar \delta$
	\end{itemize}
\end{ass}
The first item of this assumption is motivated by Assumption \ref{ass:sys}. The second item assumes that the control effort will be bounded, while the third implies that the actuator is strong enough to compensate the disturbance, which is always reasonable. Item 3 of Assumption \ref{ass:sys1} implies some approximate knowledge on the bound of the perturbation. Nevertheless, this is the worst possible disturbance that can be tolerated. As such, the overestimation of the gain might still be a problem which has to be addressed.
\begin{rem}\label{rem:tasks}
	Given a control law $u=\eta(t,x)$ that solves the PPP under Assumption~\ref{ass:sys}, the analysis of the PPPS under Assumption \ref{ass:sys1} leads to the fulfillment of one of three distinct tasks:
	\begin{enumerate}
		\item Discretization: Given $(\varepsilon,c_1)$ find $\tau$ such that $u=\eta(t,x)$ solves the PPPS.
		\item Design: Given $(\tau, c_1)$, find the minimal acceptable $\varepsilon$ so that $u=\eta(t,x)$ solves the PPPS.
		\item Feasibility: Given $(\varepsilon,\tau)$, find the minimal acceptable $c_1$.
	\end{enumerate}
Note that the solution to these tasks implies a relation between the values $(\varepsilon,\tau, c_1)$.
\end{rem}

\begin{rem} Problem \ref{prob:PPPS} modifies the conditions of Problem \ref{prob:PPP} by assuming a finite maximum admissible value for the perturbation $\delta(t,x)$. Considering a finite control authority is not done just for the sake of taking into account saturation of the actuator but rather as a way to bridge the gap between the theoretical impossibility of solving the PPP under uniform sampling and the applications reported in the literature. Furthermore, considering a maximum actuator capacity does not require knowledge of the tightest bound of the perturbation but rather only of its maximum worst-case value $\bar\delta$. The rationale is that if the tightest perturbation bound $\bar\vartheta$ is actually much lower than $\bar\delta$, then the control action will adapt to this lower bound without without overestimation. 
\end{rem}

\subsection{PPPS and BFASMC}

Next, the control law $u=\eta(t,x) = \kappa(x)$ given in (\ref{eq:BFA}) is considered. In this case, the state space can be divided into four sections see Fig. \ref{fig:div}

%(Div. 1) $|x|>\varepsilon$: Stability of the closed-loop is not ensured, (Div. 2) $\varepsilon>|x|\geq\frac{c_1}{c_1+1}\varepsilon$: $|u|\geq c_1$ and saturation occurs, (Div. 3) $\frac{c_1}{c_1+1}\varepsilon>|x|\geq\frac{\bar \delta}{1+ \bar \delta}\varepsilon$: $|u|\geq |\delta|$, (Div.4) $|x|\leq\frac{\bar \delta}{1+ \bar \delta}\varepsilon$: a relation between $u$ and $\delta$ is not ensured.

These divisions can be seen in Figure \ref{fig:div}. Note the horizontal symmetry of the proposed divisions. In the sequel we will consider $x_k>0$ without loss of generality, which is motivated from the symmetry of the problem.

\begin{figure}
	\begin{center}
\includegraphics[width=\linewidth]{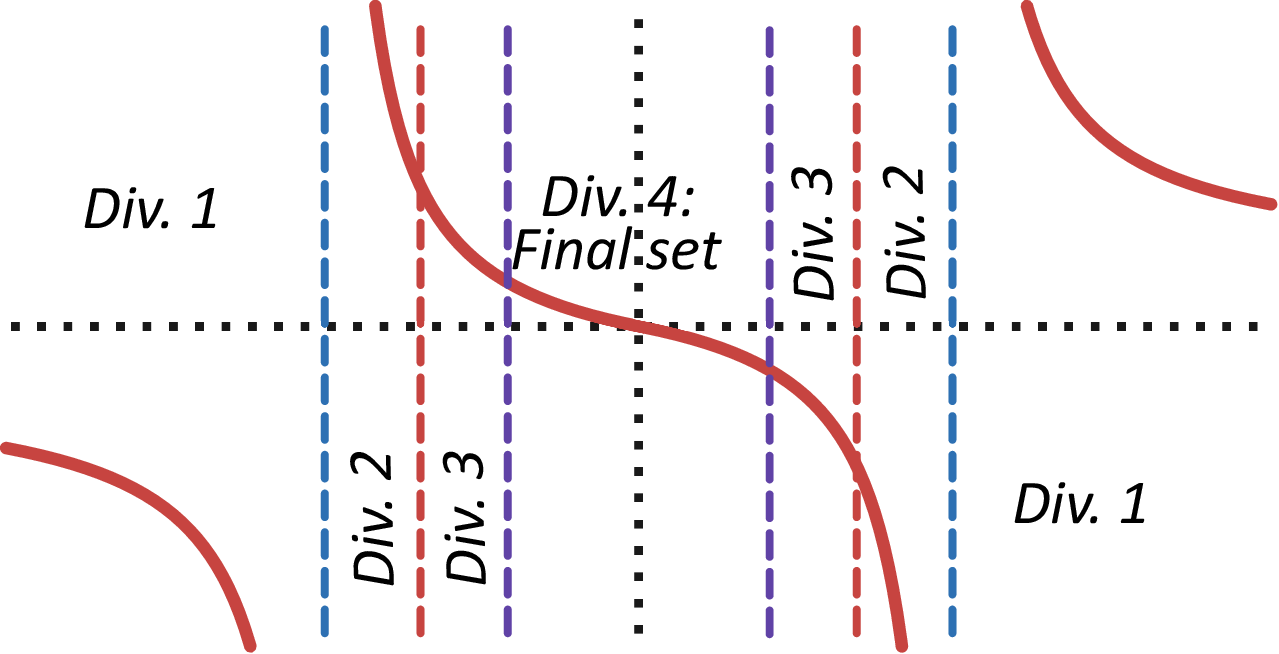}
\caption{Divisions of the state-space: [Div. 1] unstable region [Div. 2] saturation region [Div. 3] region where $|u|>|\delta(t,x)|$  [Div. 4] region where nothing can be ensured but might serve as a final set }\label{fig:div}
\end{center}
\end{figure}

%The rest of this section will aim to analyze the behavior of the BFASMC as a solution to the PPPS. As such, let us consider the following assumption
\begin{ass}\label{ass:ic}
$x_0\in$ Div. 3$\cup$ Div. 4.
\end{ass}
\begin{rem}
This assumption is analogous in this case to Assumption \ref{ass:icc} in continuous time while also avoiding the region of saturation. Thus, it will be proven that, for any initial condition which does not saturate the actuators, the solutions of the system can render some final set positively invariant, \textit{i.e.} controller \eqref{eq:BFA} under sampling solves the PPPS.
\end{rem}
Let us consider the following proposition
\begin{prop}\label{prop:sppp}
	Consider the system \eqref{eq:sys} with the control law \eqref{eq:BFA} implemented through zero-order hold under uniform sampling with period $\tau$. Let Assumptions \ref{ass:sys1} and \ref{ass:ic} hold. If $\tau<\varepsilon\left(\frac{1}{g_2c_1+1}\right)^2$, then, Div. 4 is at least asymptotically attractive and invariant.
\end{prop}
\begin{proof}
Without loss of generality, consider $x_k>0$ (this is possible from the horizontal symmetry in Figure \ref{fig:div}). Therefore, the closed-loop system can be written as: $
x_{k+1}-x_k=\int_{t_k}^{t_k+\tau}g(s))\left(\delta(s) -\frac{x_k}{\varepsilon-|x_k|}\right) ds $. The dependence on $x(t)$ of $g(t,x(t))$ and $\delta(t,x(t))$ will be omitted for the sake of readability.

 The proof is done by cases:
 
Case 1: $x_k\in$ Div. 3.
	
	In this case $c_1\geq\frac{x_k}{\varepsilon-|x_k|}\geq\bar \delta$ and, from the structure of $x_{k+1}-x_k$, $x_{k+1}-x_k\leq 0$ which means that $x_{k+1}$ must decrease. Thus, it should follow that $x_{k+1}-x_k>-g_2\tau(\bar \delta + c_1)$, this is a straight line which bounds $x(t)$ for $t\in[t_k,t_{k+1}]$. Then, one should ensure that $x_{k+1}>-\varepsilon \frac{\bar \delta}{\bar \delta +1}$,  which implies $x_{k+1}-x_k>-\varepsilon\left(\frac{\bar\delta}{1+\bar\delta}+\frac{c_1}{c_1+1}\right)$, this is the minimum admissible displacement form $x_k$ to $x_{k+1}$. Thus, one should find $\tau>0$ such that $\varepsilon\left(\frac{\bar\delta}{1+\bar\delta}+\frac{c_1}{c_1+1}\right)>g_2\tau(\bar \delta + c_1)$. Note that $\bar\delta\leq c_1$ implies the existence of a \begin{color}{red}$\vartheta>1$\end{color} such that $c_1=\vartheta\bar\delta$. Then, in this case we should ensure that \begin{color}{red} $
	g_2\tau<\varepsilon\frac{1}{(\bar \delta+1)(\vartheta\bar \delta+1)}\left(1+\frac{2\bar\delta\vartheta}{\vartheta+1}\right)$, which is feasible if the restriction  $
	g_2\tau< \varepsilon \frac{ 1}{(\vartheta\bar \delta + 1)^2}=\varepsilon \frac{ 1}{(c_1 + 1)^2}$ is imposed. \end{color}
	Furthermore, note that the condition $x_{k+1}-x_k\leq 0$ implies that the trajectories of the system must move away from the boundary of Div. 3, which means that, at least, Div. 4 is asymptotically attractive as the trajectories cannot cross Div. 3$\cup$Div. 4 completely in one sample interval if $\tau$ is chosen this way.
	
	Case 2: $x_k\in$ Div. 4. Since nothing can be directly established in this case, two subcases will be considered: 
	
	\underline{Subcase 1}: $x_k\in$Div.4\&$\int_{t}^{t+\tau}g(s)\left(\delta(s)-\frac{x_k}{\varepsilon-|x_k|}\right)ds<0.$  In this subcase we must have $x_k>x_{k+1}$, then from the fact that $x_k\in$ Div. 4, $x_k=\frac{\alpha\bar \delta}{\alpha \bar \delta +1}\varepsilon$ for some $\alpha\in[0,1)$. On the other hand, $\int_{t}^{t+\tau}g(s)\left(\delta(s)-\frac{x_k}{\varepsilon-|x_k|}\right)ds\geq -g_2\tau(1+\alpha)\bar\delta.$ 
	
	This means that the dynamics of the system can be lower bounded by a constant term and there exists a straight line which bounds $x(t)$ from below for $t\in[t_k,t_{k+1}]$. 	Then, to have $x_{k+1}\geq -\frac{\bar \delta}{\bar \delta+1}\varepsilon$, it suffices to find $\tau>0$ such that $\varepsilon\left(\frac{\alpha \bar \delta}{\alpha\bar \delta+1}+ \frac{\bar \delta}{\bar \delta +1}\right)>g_2\tau\bar \delta (\alpha+1)\implies \varepsilon \frac{1+ \alpha + 2 \alpha\bar \delta}{(\alpha\bar \delta+1)(\bar \delta+1)}>g_2\tau (\alpha+1)\implies \varepsilon\frac{1}{(\alpha\bar \delta+1)(\bar \delta+1)}\left(1+2\bar \delta \frac{\alpha}{\alpha+1}\right)>g_2\tau$  which is feasible if $ \varepsilon\frac{1}{g_2(\bar \delta+1)^2}>\tau$.
	
	\underline{Subcase 2}: $x_k\in$Div.4\&$\int_{t}^{t+\tau}g(s)\left(\delta(s)-\frac{x_k}{\varepsilon-|x_k|}\right)ds>0.$ In this subcase we must have $x_k<x_{k+1}$, then from the fact that $x_k\in$ Div. 4, $x_k=\frac{\alpha\bar \delta}{\alpha \bar \delta +1}\varepsilon$ for some $\alpha\in[0,1)$. Thus, $\frac{\bar \delta}{\bar \delta+1}\varepsilon-\frac{\alpha\bar \delta}{\alpha\bar \delta + 1}\varepsilon>x_{k+1}-x_k$.
	On the other hand, consider $\int_{t}^{t+\tau}g(s)\left(\delta(s)-\frac{x_k}{\varepsilon-|x_k|}\right)ds\leq g_2\tau(1+\alpha)\bar\delta.$ This means that the dynamics of the system can be upper bounded by a constant term and there exists a straight line which bounds $x(t)$ from above for $t\in[t_k,t_{k+1}]$.  Then, to have $x_{k+1}\leq \frac{\bar \delta}{\bar \delta+1}\varepsilon$, it suffices to find $\tau>0$ such that $\varepsilon\left( \frac{\bar \delta}{\bar \delta +1}-\frac{\alpha \bar \delta}{\alpha\bar \delta+1}\right)>g_2\tau\bar \delta (\alpha-1)\implies \varepsilon \frac{ 1-\alpha}{(\alpha\bar \delta+1)(\bar \delta+1)}>g_2\tau (1-\alpha)\implies \varepsilon\frac{1}{(\alpha\bar \delta+1)(\bar \delta+1)}\left(1+2\bar \delta \frac{\alpha}{\alpha+1}\right)>g_2\tau$  which means $ \varepsilon\frac{1}{g_2(\bar \delta+1)^2}>\tau$.

Notice that both cases consider straight lines which bound $x(t)$ for $t\in[t_k,t_{k+1}]$. This means that the arguments do not only consider the behavior of the system at $t=t_{k+1}$, but rather, the behavior of $x(t)$, for any $t\in[t_k,t_{k+1}]$, is ensured.

Then, assuming that $c_1>\bar \delta$ and that the exact value of $\bar \delta$ is not known, one can always choose $\tau< \varepsilon\frac{1}{(\vartheta\bar\delta+1)^2}=\varepsilon\frac{1}{(c_1+1)^2}$ to ensure that Div. 4 is a positively invariant set.
\end{proof}

\begin{rem}Proposition \ref{prop:sppp} characterizes a relation between $\varepsilon$, $\tau$ and $c_1$. As such, to solve any of the tasks mentioned in Remark \ref{rem:tasks}, it is only necessary to solve for the corresponding variable.
Furthermore, this relation  ensures that  Div. 4 is in fact a final set. Thus, the distance between the boundary of the barrier set and the ultimate bound on the state is bounded away from zero. This implies the boundedness of the control law and thus, that the second item of Assumption \ref{ass:sys1} is always fulfilled.
\end{rem}
\begin{rem}
The necessity of both bounds on the input gain is displayed. If $g_1\to0$, the perturbation cannot be coupled to the control input and the compensation of $\delta$ by means of a bounded control is not possible. If $g_2\to\infty$, then there will not exist some finite value of $\tau$ which ensures that the predefined behavior of the system can be kept, see Lemma 1 in \cite{hernan}.
\end{rem}

\begin{rem}
The condition $|\delta(t,x)|\leq \bar\delta$ can be replaced by $\left|\int_{t}^{t+\tau}\delta(s,x(s))ds\right|\leq \bar\delta\tau$ for all $t$, and the proof follows without change. However, this would only mean that $x_{k+1}\in\mathcal{E}$ and nothing would be ensured between these instants, \textit{i.e.} $\mathcal{F}$ would not be rendered invariant. 
\end{rem}
Controller \eqref{eq:BFA}, under the PPPS, displays three main features: (I)It avoids chattering, as $u(0)=0$. (II)It allows the predefinition of ultimate performance. (III)It avoids the overestimation of the gain used for compensating disturbances.
That is to say, considering the PPPS, the proposed selection of $\tau$ allows for all of the main features of BFASMC to be recovered under sampling.

The presented analysis implies that any practical application of \eqref{eq:BFA} under sampling must in fact consider Assumption \ref{ass:sys1} instead of Assumption \ref{ass:sys}. The difference between these assumptions is subtle, and even more so in a practical application, as any real-world problem cannot consider the compensation of bounded disturbances with an arbitrary bound. The difference in the underlying Assumptions is what causes the disparity between the PPP and the practical applications.
\begin{rem} 
In continuous time, \cite{Diego} shows that it is always possible to steer the trajectories of the system, in a predefined time, to an arbitrary neighborhood of the origin by means of a controller with a time-dependent gain which grows with time. However, since a finite control authority is assumed, this argument needs to be modified because the limited control authority actually restricts the maximum velocity of the trajectories of the system.
\end{rem}

\section{Modifying the BFA}
Up until this point the initial conditions of the system have been considered under Assumption \ref{ass:ic}. Nonetheless nothing has been said about a mechanism which drives the trajectories of the system towards this set.

In the continuous time setting, \cite{Diego} shows that it is always possible to steer the trajectories of the system to an arbitrary neighborhood of the origin by means of a proportional navigation feedback-based controller. An important drawback of this argument is the fact that, under sampling, the convergence of this algorithm cannot be ensured without knowledge of a bound of the initial condition \cite{hernan}. Furthermore, the assumption of a finite control authority limits the maximal velocity with which the trajectories of the system are attracted to the desired set.

Then, one can always use some bounded controller to ensure the finite-time attractiveness of Div. 3, at which time one can consider a new initial time. Thus, a modification to \eqref{eq:BFA} is presented to fulfill the second item of Assumption \ref{ass:sys1}.

\begin{equation}\label{eq:BFSAT}
	u=\left\{\begin{array}{cc}
		-c_1\text{sign}(x)& x\in \text{Div 1.}\cup\text{Div. 2.}\\
		-\frac{1}{\varepsilon-|x|}x&x\in \text{Div 3.}\cup\text{Div. 4.}
		
	\end{array}\right.
\end{equation}

Equation \eqref{eq:BFSAT} represents a scheme where a constant signal is used to drive the trajectories of the system to Div. 3 as fast as possible under the considered conditions. 

Then, the following theorem deals with convergence and boundedness of the trajectories of \eqref{eq:sys}-\eqref{eq:BFSAT}.
\begin{theo}\label{theo:sppp}
	Consider \eqref{eq:sys}-\eqref{eq:BFSAT},  under uniform sampling and let Assumption \ref{ass:sys1} hold. If there exists some $\gamma>0$ such that $|x_0|\leq \gamma$, and $\tau<\varepsilon\left(\frac{1}{g_2c_1+1}\right)^2$, then after a finite number of samples, $l$, it follows that $x_k\in$Div. 3 for any $k>l$. Thus, $x(t)\in$Div. 4 is achieved at least asymptotically. 
\end{theo}
\begin{proof}
	Let us assume $x_0\in\text{Div. 1}\cup\text{Div. 2,}$ otherwise the proofs follows directly from Proposition \ref{prop:sppp}.
	Then, the distance from $x_0$ to the boundary of Div. 3 is nonzero. Let us call this distance $d$.
	
	From Assumption \ref{ass:sys1}, $\bar \delta<c_1$, and so, there exists $\varsigma>0$ such that $c_1+\delta(t,x)\geq \varsigma$. This, along with the fact that the controller is in the saturation zone implies that sign$(x_{x})\int_{t_k}^{t_{k+1}}g(s,x(s))\left(\frac{\delta(s,x(s))}{\text{sign}(x_k)}-c_1\right)ds\leq-\varsigma\text{sign}(x_k)$, meaning that $x(t)$ moves away form $x_k$ and towards Div. 3. 
	
	The proof of Proposition \ref{prop:sppp} shows that $\tau<\varepsilon\left(\frac{1}{g_2c_1+1}\right)^2$ implies $|x_{k+1}-x_k|<\varepsilon\left(\frac{\bar\delta}{1+\bar\delta}+\frac{c_1}{c_1+1}\right)$, which means that the trajectories of the system cannot move through Div. 3 and past Div. 4 in one sampling interval.
	
	Therefore, after a number of samples no more than $l=\lfloor \frac{d}{\varsigma}\rfloor+1$, $x_l\in$ Div. 3$\cup$Div. 4 must be fulfilled. Consider $x_{l+1}$ as a new initial time such that $x_{l+1}$ fulfills Assumption \ref{ass:ic}. Then, the proof follows from Proposition \ref{prop:sppp}.
		\end{proof}
In this paper there are issues such as measurement noise and quantization that have not been taken into account and might cause the trajectories of the system to leave the barrier. Nonetheless, \eqref{eq:BFSAT} will not change signs, unlike \eqref{eq:BFA}, and will ensure that the system trajectories can re-converge, without the need to wait for another control law to converge.

\section{Simulation results}
In order to illustrate the results given in the last three sections, several simulations will be shown.
To approximate the continuous time behavior of the plant, a two loop discretization scheme is proposed wherein the plant will be sampled at a steady rate of $1\times10^{-6}$ seconds, while the controller is sampled at a slower rate. 

%\begin{equation*}
%	\begin{array}{rl}
%	\delta(t,x)=&d_b (0.7\cos(10t)+0.3\text{ sign}(\cos(\sqrt{2}t)))\\
%	g(t,x)=&g_1+(g_2-g_1)\frac{1+\text{sign}(\sin (3\pi t))}{2}
%	\end{array}
%\end{equation*}
%with $g_1=1$, $g_2=2$ and $\bar\delta=d_b$. The bound of the perturbation is left as a parameter so that it can be changed from simulation to simulation so as to showcase some issues. All simulations consider a value of $\varepsilon=0.01$ and an initial condition of $x_0=\varepsilon/2=0.005$.
\subsection{Unsolvability of the PPP under sampling}

To show that there always exists a finite value of $\delta(t)$ which makes $x(t)\not\in\mathcal{E}$ after one sampling interval, fix $\tau=0.01$ and $g_2=g_1=1$. Then, according to Proposition \ref{prop:del}, one can consider the value $\delta(t)=\varpi\text{sign}(x_{k-1})$, with $\varpi=\frac{\varepsilon-x_0}{\tau}=1.5$, such that the trajectories of the system leave the predefined set in one sampling interval. This is shown in Figure \ref{fig:pert}.

\begin{figure}
	\begin{center}
		\includegraphics[width=\linewidth]{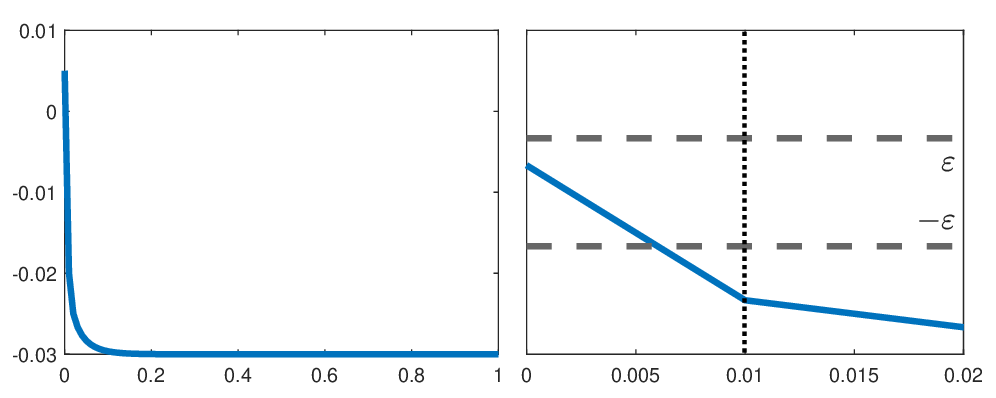}
		\caption{Simulations for a perturbed integrator with $\delta(t)$ as in Proposition \ref{prop:del}. (left) behavior of the system for 1 second (right) behavior of the system for two sampling intervals; the solution leaves the predefined set before the first sampling interval}\label{fig:pert}
	\end{center}
\end{figure}

%All of these results coincide with the theoretical predictions made in Section IV, which motivates the introduction of the PPPS.

\subsection{Behavior of sampled BFA}
It will be shown that, under the conditions given by the PPPS, Proposition \ref{prop:sppp} means that the BFASMC not only ensures that the system trajectories are constrained to the barrier set, but that the trajectories of the system lie within a final set defined by Div. 4 and that this fact is achieved with a bounded control signal if $\tau$ is chosen appropriately. 

The parameters considered within these simulations are: $\varepsilon=0.01$, $x_0=0.005$, $c_1=5$, the perturbation and input gain are given by: 
\begin{equation*}
\begin{array}{rl}
	\delta(t,x)=&\bar\delta (0.7\cos(10t)+0.3\text{ sign}(\cos(\sqrt{2}t)))\\
	g(t,x)=&g_1+(g_2-g_1)\frac{1+\text{sign}(\sin (3\pi t))}{2}
\end{array}
\end{equation*}
with $g_1=1$, $g_2=1.5$ and $\bar\delta=4.4$.
%The simulation considers a value of $\varepsilon=0.01$ and an initial condition of $x_0=\varepsilon/2=0.005$.

Then,  the selection of $\tau$ should satisfy $\tau\leq 1.44\times 10^{-4}$. For this case a value of $\tau= 1.38\times 10^{-4}$ is considered.

Figure \ref{fig:samg1} shows the behavior of the system. Not only are the trajectories of the system confined to the barrier set, but the final set is positively invariant, \textit{i.e.} the trajectories cannot escape this set once they have entered it; additionally, the control effort has no chattering and  not only does it remain bounded, but the saturation region is avoided.

\begin{figure}
	\begin{center}
		\includegraphics[width=\linewidth]{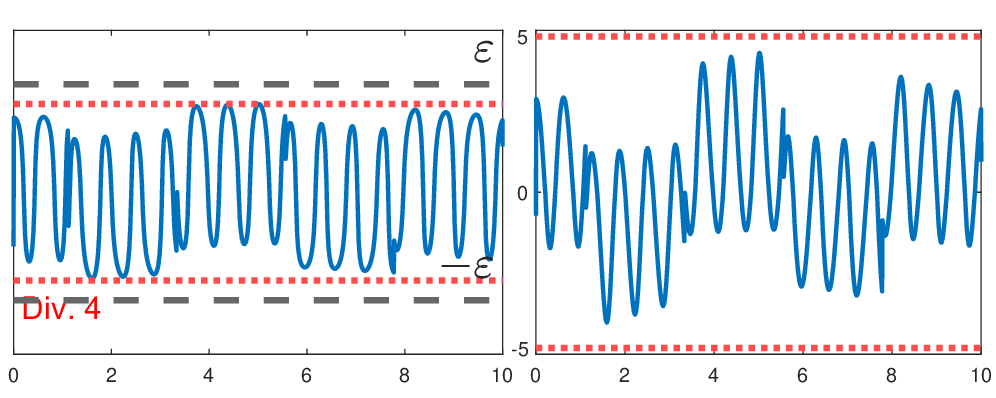}
		\caption{Response of a first order system under sampling. (left) The trajectories of the system are kept within Div. 4 (right) The controller achieves the task without crossing the limits for the control effort }\label{fig:samg1}
	\end{center}
\end{figure}

To show the advantages of the BFASMC against a linear controller, consider $g_1=1$, $g_2=1$, $\varepsilon=0.5$, $\tau=0.0062$, $c_1=7.5$ and $d_b=3$. Then, the controller $u_k=-kx_k$, with $k=17$, such that for $x_k\in\partial$Div. 3 $u_k=c_1$ will be considered.

Figure \ref{fig:lvb} shows the comparison between the linear controller and the BFASMC. The right plot shows the states; while the controller with constant gains achieves a better precision, both controllers render Div. 4 positively invariant and ensure the predefined ultimate performance, which is the goal of any predefined performance controller. The left plot shows the control effort; the controller with constant gains uses a bigger control effort to solve the PPPS, and tracks the perturbation more closely. Let us rewrite \eqref{eq:BFA} as $u=-k(x)x$ such that the similarities to the linear controller become clearer. Then, Figure \ref{fig:gains} show the magnitude of the nonlinear gain, $k(x)$, compared to the fixed gain of the linear controller. Note that the BFASMC gain is always less than half than the constant gain, using only the actuator capacity for the design, \textit{i.e.} no explicit knowledge of the perturbation is required. Thus, the issue of gain overestimation of the linear controller is solved, at least partially, by the BFASMC. 

%A clear trade-off is being displayed. while the linear controller achieves a better precision, this is done at the expense of a larger control effort. Then, while designing some application, one should consider what the priority is. If precision is more important, one can consider $\bar\delta=c_1$ in a fixed-gain design and achieve a better result, without solving the gain overestimation problem. Nonetheless if only the restriction $|x(t)|<\varepsilon$ is required, one can consider the BFASMC design and achieve the task while requiring a smaller control effort.

\begin{figure}
	\begin{center}
		\includegraphics[width=\linewidth]{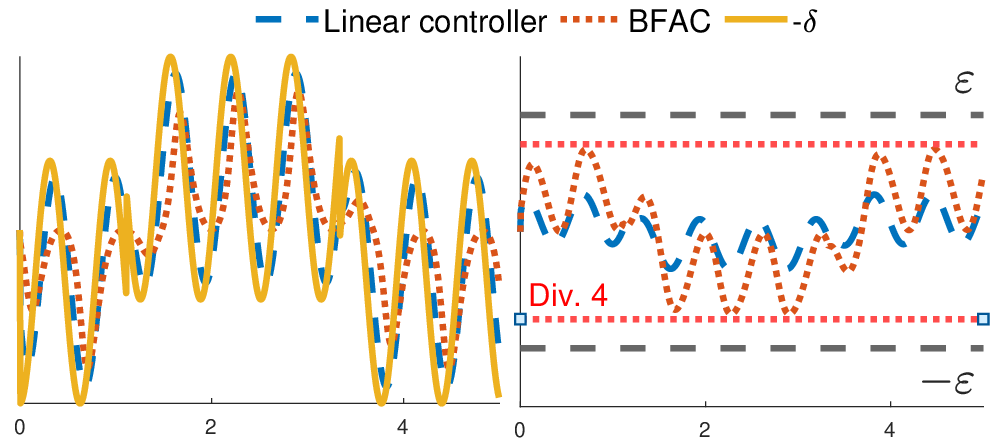}
		\caption{Comparison between a saturated linear constant gain controller and the BFASMC. (left) Control effort. (right) Trajectories of the system }\label{fig:lvb}
	\end{center}
\end{figure}
\begin{figure}
	\begin{center}
		\includegraphics[width=\linewidth]{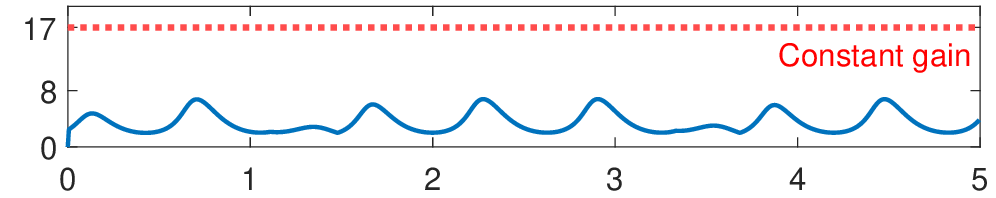}
		\caption{Comparison between the gain of the BFASMC and the constant gain of the linear controller}\label{fig:gains}
	\end{center}
\end{figure}

%To better discuss the effects of the control coefficient in the selection of $\tau$, consider the values $g_1=1$ and $g_2=1.06$ are chosen. This control coeffient varies only by $6\%$ of its maximum value. Then Figure \ref{fig:samg2} shows two cases in the top plot, it is assumed that the variation is negligible and the maximum value is not consider, while the bottom plot considers the variation in the selection of $\tau$. This shows that some knowledge of the upper bound is needed to ensure the stability of the system. Note that, since the selection is given as an inequality of this value, it is always possible to use some overestimated prediction, albeit at the cost of a faster minimum sampling time.
%
%\begin{figure}
%	\begin{center}
%		\includegraphics[width=\linewidth]{FD/sam_g2}
%		\caption{Effects of the value of $g_2$ in the response of the system. (top) The value is underestimated by $6\%$. (bottom) the value is selected correcly }\label{fig:samg2}
%	\end{center}
%\end{figure}
\subsection{Finite-time convergence to the barrier set}
To show the applicability of the proposed reaching controller, simulations were carried out considering two different initial conditions $x_0=\{0.201, 201\}$. Figure \ref{fig:ic} shows that, regardless of the different values of these initial conditions, the trajectories of the system are driven in a finite time  to Div. 3, in which case the analysis and comments made apply. 
\begin{figure}
	\begin{center}
		\includegraphics[width=\linewidth]{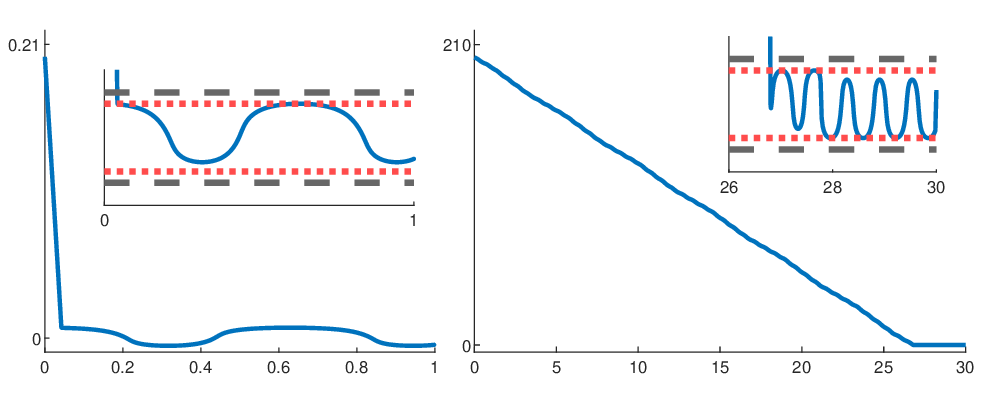}
		\caption{Trajectories with different initial conditions (left) $x_0=0.201$ (right) $x_0=201$}\label{fig:ic}
	\end{center}
\end{figure}
\section{Conclusions}
This paper has analyzed the behavior of predefined performance controllers, particularly BFASMC, under sampled-data implementations. It has been demonstrated that the conventional continuous-time theoretical framework, which assumes instantaneous feedback and infinite control authority, does not guarantee a predefined ultimate bound when applied to sampled systems. To reconcile theory with practical constraints, a revised framework has been proposed, explicitly incorporating actuator capacity limits, sampling intervals, and disturbance bounds. This reformulation establishes a critical trade-off between the sampling rate, the size of the predefined ultimate bound, and the maximal admissible disturbance, providing a structured tuning methodology for controllers in real-world applications.

Within this framework, the empirically observed success in digital implementations of BFASMC has been theoretically justified by demonstrating how its barrier-function adaptation mechanism inherently regulates gain magnitudes, mitigates chattering, and compensates for sampling-induced open-loop effects. Furthermore, a novel finite-time reaching law has been introduced, leveraging its structural properties under uniform sampling to ensure predefined-time convergence to a positively invariant final set.

A key contribution of this work is the derivation of an explicit relationship between the actuator capacity, the sampling interval, and the width of the barrier function. This result effectively provides a systematic tuning procedure for barrier-function-based sliding-mode controllers in sampled-data implementations, addressing a critical gap in the literature. This is the first time such a tuning methodology has been formally established, offering a principled way to design these controllers based on known system constraints.

These contributions collectively bridge the gap between classical predefined performance theory and practical sampled-data implementations, laying the foundation for future developments in barrier-function-based sliding-mode control for digital and embedded systems.

\bibliographystyle{ieeetran}
\bibliography{bibliografia_adaptable}
\end{document}